\documentclass[10pt, conference, compsocconf]{IEEEtran}
\pdfoutput=1

\usepackage{graphicx}
\usepackage{hyperref}
\usepackage{cite}
\usepackage{wrapfig}
\usepackage{amsfonts}

\def\Real{\ensuremath{\Bbb R}}

\def\symdiff{\mathbin{\triangle}}

\newtheorem{theorem}{Theorem}
\newtheorem{proposition}[theorem]{Proposition}
\newtheorem{lemma}[theorem]{Lemma}

\begin{document}

\title{Planar Voronoi Diagrams for  Sums of Convex Functions,\\ Smoothed Distance, and Dilation} 

\author{\IEEEauthorblockN{Matthew Dickerson}
\IEEEauthorblockA{Department of Computer Science\\
Middlebury College\\
Middlebury, Vermont, USA\\
Email: dickerso@middlebury.edu}
\and
\IEEEauthorblockN{David Eppstein}
\IEEEauthorblockA{Computer Science Department\\
University of California, Irvine\\
Irvine, California, USA\\
Email: eppstein@uci.edu}
\and
\IEEEauthorblockN{Kevin A. Wortman}
\IEEEauthorblockA{Department of Computer Science\\
California State University, Fullerton\\
Fullerton, California, USA\\
Email: kwortman@fullerton.edu}
}

\maketitle

\begin{abstract}
We study Voronoi diagrams for distance functions that add together two convex functions, each taking as its argument the difference between Cartesian coordinates of two planar points. When the functions do not grow too quickly, then the Voronoi diagram has linear complexity and can be constructed in near-linear randomized expected time. Additionally, the level sets of the distances from the sites  form a family of pseudocircles in the plane, all cells in the Voronoi diagram are connected, and the set of bisectors separating any one cell in the diagram from each of the others forms an arrangement of pseudolines in the plane. We apply these results to the smoothed distance or biotope transform metric, a geometric analogue of the Jaccard distance whose Voronoi diagrams can be used to determine the dilation of a star network with a given hub. For sufficiently closely spaced points in the plane, the Voronoi diagram of smoothed distance has linear complexity and can be computed efficiently. We also experiment with a variant of Lloyd's algorithm, adapted to smoothed distance, to find uniformly spaced point samples with exponentially decreasing density around a given point.
\end{abstract}

\begin{IEEEkeywords}
biotope transform metric;
convex function; dilation; Lloyd's algorithm;
pseudocircle; pseudoline; randomized algorithms
smoothed distance; Voronoi diagram.
\end{IEEEkeywords}

\section{Introduction}

Any bivariate function $f$ and finite set of point sites $p_i$ in the plane give rise to a \emph{minimization diagram}  in which the cell for site $p_i$ consists of points $q$ such that the value of the translated function $f(q-p_i)$ is less than or equal to the value of any of the other translates of $f$. A familiar example is given by Euclidean Voronoi diagrams, the minimization diagrams of translates of the convex functions $f(x,y)=\sqrt{x^2+y^2}$ (or, equivalently, $f(x,y)=x^2+y^2$) that measure the (squared) distance of $(x,y)$ from the origin.

Minimization diagrams have many applications, and typically have quadratic complexity~\cite{HalSha-DCG-94}, but in some important cases their complexity is much smaller.
In a Euclidean Voronoi diagram, each cell is a convex polygon, so the Voronoi diagram for $n$ sites partitions the plane into $n$ connected regions and has $O(n)$ vertices and edges. These combinatorial facts form the basis of efficient algorithms for constructing Euclidean Voronoi diagrams and using them in other geometric algorithms. It is natural, therefore, to ask: Which other minimization diagrams have connected cells and a linear number of number of features?

A partial answer was provided by Chew and Drysdale~\cite{CheDry-SoCG-85}:  Voronoi diagrams for \emph{convex distance functions} have connected cells and linear complexity. A convex function $f$ is a convex distance function if, for any positive scalar $\alpha$ and any point $p$, $f(\alpha p)=\alpha\,f(p)$. Not every convex function has this property: it implies that the \emph{level sets} $\{p\mid f(p)=k\}$ are all similar, and that the function is linear on rays from the origin, neither of which are true for all convex functions. Another answer is given by the \emph{abstract Voronoi diagrams}~\cite{Kle-89}  defined by a family of bisector curves that are required to intersect each other finitely many times and form simply connected cells. Bregman Voronoi diagrams~\cite{NieBoiNoc-SODA-07} fall into this class: they have linear bisectors and convex-polygon cells. Abstract Voronoi diagrams may be constructed efficiently~\cite{Kle-89,KleMehMei-CGTA-93,MehMeiODu-DCG-91} but it is unclear how to tell whether a given convex function has minimization diagrams that form abstract Voronoi diagrams.

In this paper we study minimization diagrams for another class of convex functions, different from the convex distance functions. The functions we study take the form $f(x,y)=g(x)+h(y)$, where $g$ and $h$ are triply-differentiable and $g'(x)g'''(x) < g''(x)^2$ and $h'(y)h'''(y) < h''(y)^2$ for all $x$ and $y$. For example, squared Euclidean distance has this form with $g(x)=x^2$ and $h(y)=y^2$. More generally, $f(x,y)=|x|^c+|y|^c$ (for $c>1$) satisfies these requirements,\footnote{The divergence of the double derivative of $|x|^c$ at the origin when $c<2$, and its non-differentiability there when $2<c<3$, do not present any serious difficulties to our theory.} and its minimization diagrams are the Voronoi diagrams for $L_c$ distance, known to have linear complexity~\cite{CheDry-SoCG-85}. We show in Section~\ref{sec:mdcf} that the minimization diagrams of arbitrary functions in this class have analogous properties:
\begin{itemize}
\item Any two \emph{level sets} $S_r(p)=\{q\mid f(q-p)=r\}$ are simple closed curves that intersect in at most two points; if they intersect at two points, they cross properly at these points. That is, these sets, which are defined analogously to Euclidean circles, form a family of \emph{pseudocircles}~\cite{KedLivPac-DCG-86,LinOrt-BAG-05}.
\item Any \emph{bisector} $B(p,q)=\{s\mid f(s-p)=f(s-q)\}$ forms either an axis-parallel line or a simple curve that is both $x$-monotone and $y$-monotone; $f(s-p)$ and $f(s-q)$ vary unimodally along the bisector $B(p,q)$.
\item Any two bisectors $B(p,q)$ and $B(p,q')$ intersect in at most one point; if they intersect, they do so at a proper crossing. That is, if we fix $p$ and let $q$ vary, the bisectors $B(p,q)$ form a \emph{weak pseudoline arrangement}~\cite{Epp-GD-04,EppFalOvc-07,FraOss-GD-03}, and the correspondence between $q$ and $B(p,q)$ can be viewed as a form of duality for these pseudolines.
\item Each cell of the minimization diagram is simply connected, as it is a single cell in a pseudoline arrangement. Thus, the minimization diagram has linear complexity and, like Voronoi diagrams for convex distance functions and abstract Voronoi diagrams,\footnote{We do not show that these diagrams satisfy all requirements of an abstract Voronoi diagram, as we do not bound the number of crossings between bisectors for unrelated pairs of points.} can be constructed in randomized expected time $O(n\log n)$ (Theorem~\ref{thm:voronoi}).
\end{itemize}
In Section~\ref{sec:bad-examples} we provide examples showing that our restrictions are necessary: minimization diagrams for more general convex functions do not in general have these properties.

Our motivation for studying this type of minimization diagram comes from \emph{smoothed distance}. Given a fixed point $o$ in the plane, the \emph{smoothed distance} or \emph{biotope transform metric} $d_o(p,q)=2d(p,q)/(d(p,o)+d(q,o)+d(p,q))$~\cite{Cla-UCI-08,DezDez-09} is a geometric analogue of the Jaccard similarity measure for clustering binary data.  A maximal set of points with a given minimum smoothed distance will be distributed around $o$ with exponentially decreasing density~\cite{Cla-UCI-08}, but as we show in Section~\ref{sec:lloyd}, the point spacing may be improved by using smoothed distance in Lloyd's algorithm~\cite{DuFabGun-SR-99,Llo-ITIT-82}, a continuous variant of $K$-means clustering that repeatedly moves each site to the centroid of its Voronoi cell.

Smoothed distance is a monotonic function of the \emph{dilation} $(d(p,o)+d(o,q))/d(p,q)$, a measure of the quality of a star network as an approximation to the Euclidean distances among a set of points. For the maximum dilation pair of points $(p,q)$ from a set $P$ for a fixed center $o$, $q$ is among the $O(1)$ nearest neighbors to $p$ either in Euclidean distance or in the sequence of points sorted by distance from $o$, and this result forms the basis of an efficient algorithm for finding the center $o$ that minimizes the maximum dilation~\cite{EppWor-CGTA-07} . However, this can be simplified using smoothed distance: the maximum-dilation pair must be neighbors in the smoothed distance Voronoi diagram (Proposition~\ref{prop:dilation-neighbors}).

Neither smoothed distance nor dilation are translates of convex functions. However, in Section~\ref{sec:xform} we use complex logarithms to transform the plane and we perform a suitable monotonic transformation of smoothed distance, showing that Voronoi diagrams for smoothed distance are equivalent to minimization diagrams for the convex function
$$f(x,y)=\ln\frac{(1+e^x)^2}{e^x} - \ln (1+\cos y).$$
As we require, this function is the sum of univariate functions $g(x)$ and $h(y)$, with $g'g'''<(g'')^2$. It is not true that $h'h'''<(h'')^2$ for all $y$, but it is true for $-\pi/2<y<\pi/2$, so smoothed distance Voronoi diagrams are well behaved whenever each Voronoi cell spans an angle of at most $\pi/2$ on each side of its site with respect to~$o$ (Theorem~\ref{thm:smoothed}).

\section{Dilation, Smoothed Distance, and Logarithmic Transformation}
\label{sec:xform}

Two important measures of similarity between sets $A$ and $B$ are the Hamming distance $d_H(A,B)=|A\symdiff B|$ (where $\symdiff$ denotes the symmetric difference of sets) and the Jaccard distance~\cite{Jac-BSV-01}, which weights the Hamming distance by the size of the union of the sets:
$$J(A,B)=\frac{|A\symdiff B|}{|A\vee B|}=\frac{d_H(A,B)}{d_H(A,\emptyset)+d_H(B,\emptyset)-d_H(A,B)}.$$
A monotone transformation of this distance lies in the interval $[0,1]$:
$$d_J(A,B)=\frac{2}{\frac{1}{J(A,B)}+2}
=\frac{2d_H(A,B)}{d_H(A,\emptyset)+d_H(B,\emptyset)+d_H(A,B)}.$$
The same formula defining the (modified) Jaccard distance in terms of the Hamming distance can be used to derive a new metric from any given metric space $(X,d)$ and fixed point $o\in X$. Define the \emph{$o$-smoothed distance} or \emph{biotope transform metric} by the formula
$$d_o(p,q)=\frac{2d(p,q)}{d(p,o)+d(q,o)+d(p,q)}.$$

This is a metric on $X\setminus\{o\}$, as can be shown using the tight span~\cite{Isb-CMH-64,Dre-AiM-84}, the minimal $L_\infty$-like metric completion of a metric space.
Any four points $\{o,p,q,r\}\subset (X,d)$ may be embedded isometrically into a metric space formed by an axis-aligned rectangle of the $L_1$ plane (possibly a degenerate rectangle), together with four line segments (possibly of length zero) connecting the corners of this rectangle to the four points. When the four line segments all have length zero, the four points are in cyclic order ($o$, $p$, $r$, $q$) on the rectangle's corners, and the rectangle has aspect ratio $a$, then the triangle inequality for $d_o$ is satisfied exactly: $d_o(p,q)=1=\frac{1}{a+1}+\frac{a}{a+1}=d_o(p,r)+d_o(q,r)$. Increasing the length of the line segments connecting the four points to the rectangle or placing the points in a different cyclic order only strengthens the triangle inequality.

For the points within a $d$-ball of radius $\epsilon\,d(o,c)$ around a point $c$, the factor $d(p,o)+d(q,o)+d(p,q)$ in the definition of $d_o$ ranges from $(1-\epsilon)2d(o,c)$ to $(1+2\epsilon)2d(o,c)$, varying by a factor of approximately $1+3\epsilon$ within this range as $\epsilon$ approaches zero. Thus, closely spaced sets of points in the smoothed distance have distances that are approximately similar to their unsmoothed distances. As a metric space on $X\setminus\{o\}$, the smoothed distance $d_o$ is  topologically equivalent to the metric induced by $d$ on the same space: both distances have the same open sets and neighborhood structures.

Smoothed distance $d_o$ for the Euclidean plane is invariant with respect to rotations and scaling centered at $o$. These transformations may be expressed by representing point $(x,y)$ as the complex number $x+iy$, with $o=0$: if $z$ is any complex number, then the product $zp$ may be interpreted geometrically as rotating the point $p$ by the angle $\angle z01$ and scaling the rotated point by $|z|$. With this interpretation,
\begin{wrapfigure}{r}{0.23\textwidth}
\centering\includegraphics[width=0.23\textwidth]{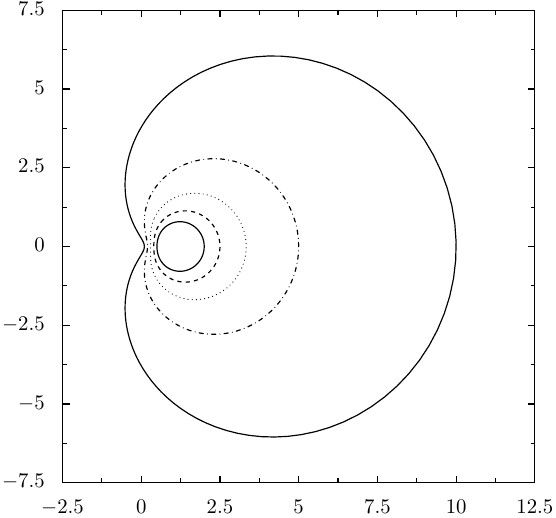}
\caption{Concentric circles for smoothed distance in the Euclidean plane, with $o=(0,0)$, $p=(1,0)$, and radii 0.5, 0.6, 0.7, 0.8, and 0.9}
\label{fig:polarcircles}
\end{wrapfigure}
$d_o(zp,zq)=d_o(p,q)$.

As Figure~\ref{fig:polarcircles}
shows, smoothed-distance circles with smaller radius closely resemble Euclidean circles, while larger-radius smoothed-distance circles are not even convex.

Smoothed distance is closely related to \emph{dilation}, a measure of the quality of a graph as an approximation to a metric space~\cite{Epp-HCG-00}. The dilation of vertices $p$ and $q$ is the ratio of their distances in the graph and in the ambient space; the dilation of the graph is the largest dilation of any pair of vertices. For a star graph $G$ with center $o$ and all other vertices as leaves~\cite{EppWor-CGTA-07}, the dilation of pair $(p,q)$ is
$$\frac{d(p,o)+d(o,q)}{d(p,q)} = \frac{2}{d_o(p,q)}-1$$
and the dilation of the whole star graph is
$$\max_{p,q\in V(G)}\frac{d(p,o)+d(o,q)}{d(p,q)} = \max_{p,q\in V(G)}\frac{2}{d_o(p,q)}-1.$$

Because dilation is a monotonic transformation of smoothed distance, the point $p$ that is nearest to a query point $q$ in terms of smoothed distance is also the point that has the greatest dilation with $q$. The Voronoi diagram for smoothed distance partitions the plane into regions such that the region containing a query point $q$ corresponds to the point $p$ for which a path from $q$ will have to make the greatest detour by passing through $o$ instead of connecting directly. The adjacency relations between cells in this Voronoi diagram are also meaningful for dilation:

\begin{proposition}
\label{prop:dilation-neighbors}
The points $p$ and $q$ defining the dilation of a star graph have adjacent cells in the Voronoi diagram for smoothed distance, using the star's leaves as Voronoi sites and its hub as the point~$o$.
\end{proposition}

\begin{proof}
Form the Voronoi diagram for $V(G)\setminus\{q\}$ and then add $q$ to form the Voronoi diagram of $V(G)$. Let $R$ be the Voronoi region of $p$ prior to adding $q$.
$R$ must contain $q$, because otherwise some other point than $p$ would define the greatest dilation with respect to $q$. After adding $q$, part of $R$ becomes incorporated into the Voronoi region for $q$, while the rest of $R$ (in particular the point $p$ itself) remains in the region of $p$. Thus, these two regions meet.
\end{proof}

Thus, we would like to understand and compute Voronoi diagrams for smoothed distance. Smoothed distance is neither convex nor translation-invariant, but we can transform it into an equivalent distance that is, by interpreting polar coordinates in the plane for which we are computing smoothed distance as Cartesian coordinates in a transformed plane. Equivalently, using complex numbers (with $o=0$), if complex number $z$ has polar coordinates $(r,\theta)$, then $\ln z$ has Cartesian coordinates $(\ln r,\theta)$. Define a distance $\delta$ on the transformed complex number plane by the formula
$$\delta(p,q)=d_o(e^p,e^q).$$
This logarithmic transformation replaces the invariance of $d_o$ under complex multiplication by invariance of $\delta$ under complex addition (that is, translation of the plane):
$$\delta(p+z,q+z)=\delta(p,q).$$
Figure~\ref{fig:XYcircles} shows concentric circles for~$\delta$.

\begin{figure}[t]
\centering\includegraphics[height=1.25in]{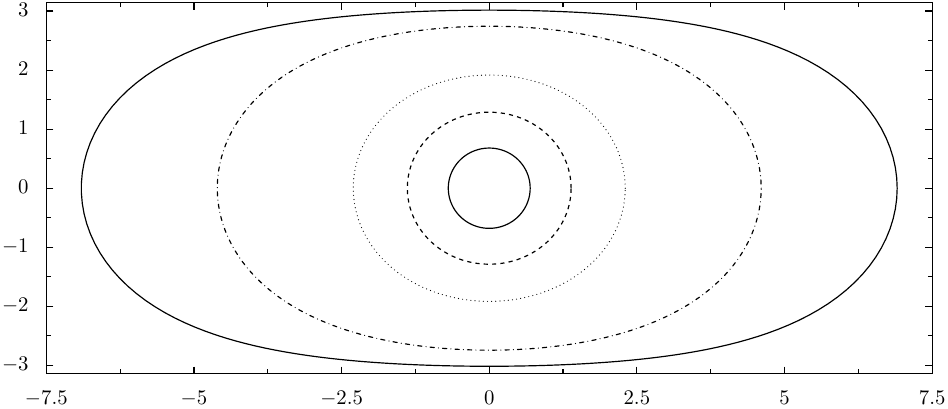}
\caption{Concentric circles for the logarithmically transformed distance $\delta$ with radii 0.5, 0.75, 0.9, 0.99, and 0.999.}
\label{fig:XYcircles}
\end{figure}

We may calculate $\delta((x,y),(0,0))$ directly as a formula of $x$ and $y$:
it equals $d_o(p,q)$, where $o=(0,0)$, $p=(e^x\cos y,e^x\sin y)$, and $q=(1,0)$.
Thus, $d(p,o)=e^x$, $d(q,o)=1$, and
\begin{eqnarray*}
d(p,q)&=&\sqrt{(e^x\cos y - 1)^2 + (e^x\sin y)^2)}\\
&=&\sqrt{(e^x)^2 - 2e^x\cos y + 1}.
\end{eqnarray*}
Therefore, the smoothed distance is
\begin{eqnarray*}
d_o(p,q)&=&\frac{2d(p,q)}{d(p,o)+d(q,o)+d(p,q)}\\
&=&\frac{2\sqrt{(e^x)^2 - 2e^x\cos y + 1}}{e^x + 1 + \sqrt{(e^x)^2 - 2e^x\cos y + 1}}.
\end{eqnarray*}
It is convenient to monotonically transform this formula:
$$\left(\frac{2}{d_o(p,q)}-1\right)^2=\frac{(e^x)^2+2e^x+1}{(e^x)^2 - 2e^x\cos y + 1}$$
and
$$\frac12\left(1-1/\left(\frac{2}{d_o(p,q)}-1\right)^2\right)=\frac{(\cos y + 1)e^x}{(e^x)^2+2e^x+1}.$$
Therefore, if we define
$$f(x,y)=-\ln\left(\frac12\left(1-1/\left(\frac{2}{d_o(p,q)}-1\right)^2\right)\right),$$
it follows that
$$f(x,y)=\ln\frac{(1+e^x)^2}{e^x} - \ln(\cos y + 1).$$
Thus, we have represented smoothed distance as a monotonic transform of  translates of $f(x,y)=g(x)+h(y)$, where
$$g(x)=\ln\frac{(1+e^x)^2}{e^x}
\mbox{~~and~~}
h(y)= - \ln(\cos y + 1).$$
Graphs of these two functions are depicted in Figure~\ref{fig:components}.

\begin{figure}[t]
\centering\includegraphics[width=1.5in]{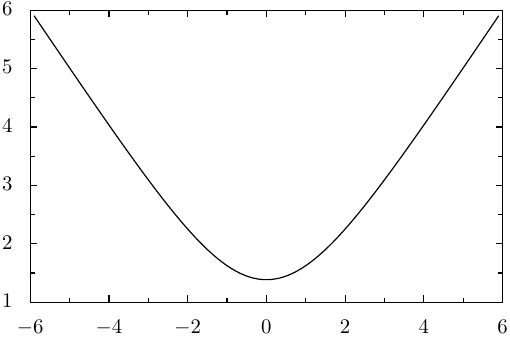}
\qquad\includegraphics[width=1.5in]{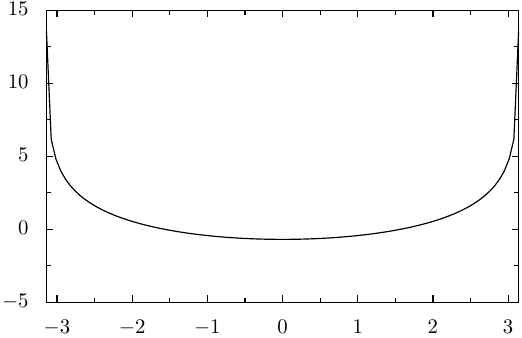}
\caption{Graphs of the functions $g(x)$ and $h(y)$ used to represent the transformed value of the smoothed distance.}
\label{fig:components}
\end{figure}

\section{Minimization Diagrams of Convex Functions}
\label{sec:mdcf}

The transformation in Section~\ref{sec:xform} from a Voronoi diagram of smoothed distance to a minimization diagram of translates of $f(x,y)=g(x)+h(y)$ motivates the more general study of minimization diagrams of this type of function.
A univariate function $f : \Real\mapsto\Real$ is \emph{convex} if the set of points 
\[ \{ (x, y) \mid y \geq f(x) \} \]
on or above $f(x)$ is a convex subset of the plane.  In higher dimensions, a multivariate function $f:\Real^d\mapsto\Real$ is convex if the set
$$\bigl\{(x_0,x_1,x_2,\dots x_d)\mid f(x_0,x_1,x_2,\dots x_{d-1})\le x_d\bigr\}$$
is a convex subset of $\Real^{d+1}$. 
A convex univariate function is \emph{strictly convex} if there is no interval within which it is linear. For a doubly-differentiable univariate function $g(x)$, convexity is equivalent to the inequality $g''\ge 0$ and strict convexity is equivalent to the inequality $g''>0$. A curve in the $(x,y)$-plane is \emph{$x$-monotone} if every line parallel to the $y$ axis intersects it in at most one point; intuitively, these curves are the graphs of functions from $x$ to $y$. Symmetrically, a curve is \emph{$y$-monotone} if every line parallel to the $x$ axis intersects it in at most one point. These concepts should be distinguished from that of a function being \emph{monotonically increasing} or \emph{strictly monotonically increasing}: if $x_1>x_0$, and $\phi(x)$ is monotonically increasing, then $\phi(x_1) \ge \phi(x_0)$; if  $\phi(x)$ is strictly monotonically increasing, then $\phi(x_1) > \phi(x_0)$. In a minimization diagram of translates of a function $f(x,y)$, a \emph{bisector} $B(p,q)$ of two distinct points $p$ and $q$ is the locus of points with equal values of $f_p$ and $f_q$: $B(p,q)=\{r\mid f_p(r)=f_q(r)\}$. The minimization diagram of the points $\{p,q\}$ is formed by partitioning the plane into two cells along the bisector.

\begin{proposition}
\label{prop:monotone-bisectors}
Let $f(x,y)=g(x)+h(y)$, where $g$ and $h$ are strictly convex. Then any bisector $B(p,q)$ is either an axis-parallel line or an $x$-monotone and $y$-monotone curve.
\end{proposition}

\begin{proof}
If $p$ and $q$ have equal $x$-coordinates, then $f_p(r)=f_q(r)$ whenever $h(r_y-p_y)=h(r_y-q_y)$; since this condition is independent of $x$, the bisector must be the union of one or more lines parallel to the $y$ axis. Otherwise, any line $x=C$ parallel to the $y$ axis intersects the bisector $B(p,q)$ in the points $r$ such that $h(r_y-p_y)-h(r_y-q_y)=g(C-q_x)-g(C-p_x)$. By strict convexity, the function $H(y)=h(r_y-p_y)-h(r_y-q_y)$ is strictly monotonically increasing, so there is at most one such point and the bisector is $x$-monotone.
Symmetrically, if $p$ and $q$ have the same $y$-coordinate, the bisector must be the union of one or more lines parallel to the $x$ axis, and otherwise it is $y$-monotone. The only curves that are simultaneously $x$-monotone or a union of $y$-parallel lines, and $y$-monotone or a union of $x$-parallel lines, are the ones listed in the proposition: a single axis-parallel line, or a doubly-monotone curve.
\end{proof}

In particular, each bisector must be a \emph{pseudoline} (the image of a line under a homeomorphism of the plane\footnote{This definition is from~\cite{Sho-VKF-91}; see~\cite{EppFalOvc-07} for a comparison with other common definitions of pseudolines.}), because it is either itself a line or a monotonic curve that partitions the plane into two cells. In the next sequence of lemmas, we show that the level sets of the translates of $f$ are \emph{pseudocircles} (simple closed curves that cross each other at most twice) by comparing their curvatures at tangent points of the same slope. For convenience, when the argument of the functions $g(x)$ and $h(y)$ is specified in the context, the notations $g$, $h$, $g'$, $h'$, etc., refer to the values $g(x)$, $h(y)$, $\frac{d}{dx}g(x)$, $\frac{d}{dy}h(y)$, etc.

\begin{lemma}
Let $f(x,y)=g(x)+h(y)$ where $g$ and $h$ are convex and differentiable.
Let  $p=(x_p,y_p)$ be any point other than the global minimum of $f$, let $f(p)=r$, and let $S_r$ be the level set $\{q\mid f(q)=r\}$.
Then the slope of the tangent to $S_r$ at $p$ is $-h'/g'$.
\end{lemma}

\begin{proof}
Since $p$ is not the minimum, $S_r$ is a simple closed curve.
Therefore there is a unique tangent line which is characterized by the fact that, at points $q$ on the line at a small distance $\epsilon$ from $p$, the difference between $f(p)$ and $f(q)$ vanishes to first order. A point $q$ at a distance proportional to $\epsilon$ on the line of slope $-h'/g'$ can be found by adding $(-h'\epsilon,g'\epsilon)$ to $p$; at this point, $f(q)$ is $f(p)+(-h'\epsilon)(g'+O(\epsilon))+(g'\epsilon)(h'+O(\epsilon))=f(p)+O(\epsilon^2)$, so the difference vanishes to first order as desired.
\end{proof}

Next, suppose we move $\epsilon$ units along the curve $S_r$ from $p$; to within
first order, the point we move to is found by adding
$(\epsilon\frac{-h'}{\sqrt{h'^2+g'^2}},\epsilon\frac{g'}{\sqrt{g'^2+h'^2}})$ to $p$. When we do so, $g'$
changes by a factor of
$1-\epsilon\frac{g''h'}{g'\sqrt{g'^2+h'^2}}+O(\epsilon^2)$,
and $h'$ changes by a factor of $1+ \epsilon\frac{g'h''}{h'\sqrt{g'^2+h'^2}}+O(\epsilon^2)$, so the slope $h'/g'$
changes by a factor of $1+\epsilon\frac{g'h''/h' +
h'g''/g'}{\sqrt{g'^2+h'^2}}+O(\epsilon^2)$. We may view the term $\frac{g'h''/h' +
h'g''/g'}{\sqrt{g'^2+h'^2}}$ appearing in this expression as a local measure of the
curvature of $S_r$; it is not rotation-invariant but can be used to compare
the curvatures of two curves at points of equal tangent slope. Larger
values of this term mean a smaller radius of curvature and lower values mean a greater
radius. To show that the radius of curvature at a given slope increases as $r$ increases, we will examine the behavior of this function as we increase $r$ by a small quantity $\epsilon$ while moving $p$ along a curve that keeps the slope of the tangent to $S_r$ fixed.

\begin{lemma}
Let $f(x,y)=g(x)+h(y)$ where $g$ and $h$ are convex and twice-differentiable.
At any point $p$ for which $g',h'\ne 0$, define the curve $T_p$ through $p$ to consist of points with the same value of $-h'/g'$ as at $p$. Then the tangent line to $T_p$ at $p$ has slope $\frac{h''}{h'}/\frac{g''}{g'}$.
\end{lemma}

\begin{proof}
Move $p$ by a distance proportional to some small $\epsilon$ along this line, by adding the vector $(\epsilon\frac{h''}{h'},\epsilon\frac{g''}{g'})$ to $p$. As $p$ moves, $g'$ and $h'$ both change by factors of $1+\epsilon\frac{g''h''}{g'h'}+O(\epsilon^2)$. Both factors are equal to within first order, so the slope $-h'/g'$ does not change to first order and we have identified the correct tangent line to $T_p$.
\end{proof}

\begin{lemma}
\label{lem:curvature}
Let $f(x,y)=g(x)+h(y)$ where $g$ and $h$ are convex and triply-differentiable, and where $g'''g' < (g'')^2$ and $h'''h' < (h'')^2$ for all $x$ and $y$. Then, along any curve $T_p$, the radius of curvature of the curves $S_r$ at the points where these curves are crossed by $T_p$ is a monotonically increasing function of $f(p)$.
\end{lemma}

\begin{proof}
We assume without loss of generality that $g$, $h$, $g'$, and $h'$ are all positive at $p$, for otherwise we may achieve these assumptions by translating the plane, reflecting it across one or both of the coordinate axes, or adding a constant to $f$, without changing the truth of the lemma.
As in the previous lemma, we move at a distance proportional to $\epsilon$ along $T_p$, to within first order, by adding $(\epsilon\frac{h''}{h'},\epsilon\frac{g''}{g'})$ to $p$. And as in the previous lemma, this motion causes $g'$ and $h'$ to both change by a factor of $1+\epsilon\frac{g''h''}{g'h'}+O(\epsilon^2)$; this same factor also describes the change in $\sqrt{g'^2+h'^2}$. Thus, to find the direction of change to
the value $\frac{g'h''/h' +
h'g''/g'}{\sqrt{g'^2+h'^2}}$ that we are using to compare curvatures for a given tangent slope, it only remains to evaluate the change to $g''$ and $h''$. The double derivative $g''$, and therefore the term $h'g''/g'$ appearing in the numerator of our local curvature function, changes by a factor of
$1+\epsilon\frac{g'''h''}{g''h'}+O(\epsilon^2)$; if $g'''g'<(g'')^2$, this is smaller than the factor of $1+\epsilon\frac{g''h''}{g'h'}+O(\epsilon^2)$ describing the change to the denominator of our local curvature function. The double derivative $h''$, and therefore the term $g'h''/h'$ appearing in the numerator of our local curvature function, changes by a factor of $1+\epsilon\frac{g''h'''}{g'h''}+O(\epsilon^2)$; again, if $h'''h'<(h'')^2)$, this is smaller than the factor of $1+\epsilon\frac{g''h''}{g'h'}+O(\epsilon^2)$ describing the change to the denominator of our local curvature function. Thus, if the assumptions of the lemma are met, $\frac{g'h''/h' +
h'g''/g'}{\sqrt{g'^2+h'^2}}$ decreases and the radius of curvature increases as $f(p)$ increases along $T_p$.
\end{proof}

Due to the convexity of $f$, and the strict convexity of $g$ and $h$, the level sets $S_r(p)=\{q\mid f_p(q)=r\}$ are themselves convex; in particular, they are simple closed curves in the plane and, as the following proposition shows, they form a family of \emph{pseudocircles} in the plane.

\begin{proposition}
\label{prop:pseudocircles}
Let $f(x,y)=g(x)+h(y)$ be a convex function such that $g$ and $h$ are triply-differentiable, $g'''g'<(g'')^2$, and $h'''h'<(h'')^2$, and let $r_1$ and $r_2$ be numbers and $p_1$ and $p_2$ be points such that $(r_1,p_1)\ne (r_2,p_2)$. Then the level sets $S_{r_i}(p_i)$ intersect in at most two points; if they intersect at two points, they cross properly at these points.
\end{proposition}

\begin{proof}
We suppose that a pair that forms more than two points of intersection or that has two non-crossing intersections exists, and proceed to derive a contradiction. If there are exactly two points of intersection, one of which is not a crossing, then (because both are simple closed curves) both must be points of tangency; we may increase the radius of the inner level set slightly and form two level sets that cross four times, so we may assume without loss of generality that there are three or more points of intersection. If $r_1=r_2$, then again we may change one of the radii slightly while preserving the property of having more than two points of intersection, for this change of radii cannot remove any crossings and can be chosen to turn at least half of the points of tangency into pairs of crossings.
And if some of the points of intersection are tangencies, we may increase or decrease $r_1$ by a small amount and replace at least half of the tangencies by crossings without changing the property that there are more than two intersection points. Thus, we may assume that the two level sets intersect more than two times at proper crossings.

Now, let $T$ be the set of real numbers $x>1$ such that $S_{r_1}(p_1)$ and $S_{r_2}(xp_2+(1-x)p_1)$ have more than two proper crossings; that is, we consider translating $p_2$ directly away from $p_1$ for as far as we can while preserving the overly large number of crossings between the two curves. Because both level sets are bounded, $I$ is itself bounded; let $t=\sup T$ and let $p_3=tp_2+(1-t)p_1$.
In order for the translated level sets to have more than two crossings for $x<t$ but only to have two crossings for $x>t$, the sets $S_{r_1}(p_1)$ and $S_{r_2}(p_3)$ must be tangent. However, this point of tangency cannot be one at which the two level sets cross, because our assumptions imply that both curves have curvature that varies continuously along the curves. It cannot be a tangency in which the curve with the smaller value of $r_i$ lies inside the curve with the larger value of $r_i$, because then for $x>t$ the tangency would become two crossings and $t$ would not be equal to $\sup T$. It cannot be a tangency in which the two curves meet externally, because then by convexity for $x<t$ the curves would have only two crossing points near the tangency and again $t$ would not be equal to $\sup T$. And it cannot be a tangency in which the curve with the smaller value of $r_i$ lies outside the curve with the larger value of $r_i$, because that would violate Lemma~\ref{lem:curvature}. However, these exhaust the possible ways the two curves can be tangent. This contradiction completes the proof.
\end{proof}

The next proposition states that (with the same assumptions as Lemma~\ref{lem:curvature} and Proposition~\ref{prop:pseudocircles}) the bisectors $B(p,q)$ and $B(p,s)$ act like pseudolines: they meet in at most one point, and if they meet they cross properly.

\begin{proposition}
\label{prop:pseudolines}
Let $f(x,y)=g(x)+h(y)$ be a convex function such that $g$ and $h$ are triply-differentiable, $g'''g'<(g'')^2$, and $h'''h'<(h'')^2$. Then any two bisectors $B(p,q)$ and $B(p,s)$ defined from $f$ have at most one point of intersection. If they intersect, they cross properly.
\end{proposition}

\begin{proof}
We show that $B(p,q)$ and $B(p,s)$ meet in at most one point by showing that any two points $t$ and $u$ in the plane are intersected by at most one bisector $B(p,\cdot)$. If there is a bisector $B(p,q)$ that contains both of these points, then $p$ and $q$ are equidistant from $t$ and from $s$; that is, $p$ and $q$ both belong to the level sets $S_t=\{s\mid f(t-s)=f(t-p)\}$ and $S_u=\{s\mid f(u-s)=f(u-p)\}$. These level sets are rotated by $\pi$ from the level sets of Proposition~\ref{prop:pseudocircles} but that rotation does not affect the conclusion of the proposition: they have at most two points of intersection. One of these intersection points is $p$ and the other is $q$; there can be no third intersection to form another bisector with $p$ through $t$ and $u$. If $B(p,q)$ and $B(p,s)$ met in two points, it would violate the uniqueness of bisectors through pairs of points, so such a double intersection cannot happen.

To complete the proof of the proposition, we must show that if two bisectors intersect, they meet in a proper crossing. But if two bisectors $B(p,q)$ and $B(p,s)$ met at a point of tangency without crossing, then a small translation of $s$ either towards or away from $p$ would transform this tangency into a pair of crossings, violating the first part of the proposition.
\end{proof}

In other words, if we fix $p$ and let $q$ vary, the family of bisectors $B(p,q)$ forms a weak pseudoline arrangement. However, the proposition applies only to pairs of bisectors that share one of their defining points. These results are not quite enough to show that minimization diagrams for $f$ are abstract Voronoi diagrams in the sense of Klein~\cite{Kle-89}, because abstract Voronoi diagrams require all bisectors to have a constant number of intersection points. However, the same general results as for abstract Voronoi diagrams follow in this case.

\begin{theorem}
\label{thm:voronoi}
Let $f(x,y)=g(x)+h(y)$ be a convex function such that $g$ and $h$ are triply-differentiable, $g'''g'<(g'')^2$, and $h'''h'<(h'')^2$. Then any minimization diagram for $f$, defined by a finite set $P$ of $n$ point sites, subdivides the plane into $n$ simply-connected regions, one per site. The diagram can be constructed in time $O(n\log n)$ using a primitive that finds minimization diagrams for three sites.
\end{theorem}

\begin{proof}
We may assume without loss of generality that the minima of $g$ and $h$ occur at $x=y=0$, for otherwise we may add an appropriate constant to $x$ and $y$ without changing the combinatorial description of the minimization diagram. Thus, the cell for any point $p$ contains $p$ itself. In any weak arrangement of pseudolines, any nonempty intersection of halfspaces defined by the pseudolines forms a single cell of the arrangement (e.g. see Theorem~11.4.11 of~\cite{EppFalOvc-07}).
It follows by Proposition~\ref{prop:pseudolines} that the cell of $p$ in the minimization diagram is a single cell in a weak arrangement of pseudolines and is therefore simply connected. Thus, the minimization diagram for $f$ and $P$ consists of $n$ simply connected cells.

We construct the diagram by a standard randomized incremental algorithm in which we add points to $P$ one at a time according to a uniformly random permutation, and maintain a history DAG describing the cells of the diagram in past states of the construction. We also maintain the sequence of bisectors surrounding each cell of the diagram, in a balanced binary tree data structure. To add a point $p$, we use the history DAG to locate the cell of the diagram that contains $p$. We then form a list $L$ of cells known to overlap with the cell of $p$; initially this list includes only the cell containing $p$. To build the cell for $p$, we repeatedly remove the cell for a site $q$ from $L$, and split this cell along the bisector $B(p,q)$. The points where this bisector crosses the boundary of the cell for $q$ may be found by a binary search of the boundary, in which each step consists of finding the vertex of the minimization for $p$, $q$, and a third site determining one of the boundary segments, and comparing the $x$ and $y$ coordinates of this vertex to those of the other vertices on the segment. After the part to be removed from the cell of $q$ is determined, the boundary segments in the binary search tree for $q$ are removed and the associated cells are added to $L$ if necessary.

The time to locate $p$ is $O(\log n)$ in expectation by a standard analysis of history DAGs.
Each new feature of the minimization diagram takes $O(\log n)$ time to construct, using the binary search trees, and there are in expectation $O(1)$ new features for the $i$th added site. Thus, the total expected time for the construction is $O(n\log n)$.
\end{proof}

\section{Two Bad Examples}
\label{sec:bad-examples}

The assumptions that we make on the form of $f(x,y)=g(x)+h(y)$ as a sum of two univariate convex functions, and on the triple derivatives of these functions, may seem technical and unnecessary. However, in this section we provide examples showing that without the assumption on the form of $f$ beyond convexity, its minimization diagrams may have quadratic complexity, and without assumption on the triple derivatives of $g$ and $h$, the level sets for $f$ may not be pseudocircles.

\begin{wrapfigure}{r}{0.25\textwidth}
\centering\includegraphics[height=1.5in]{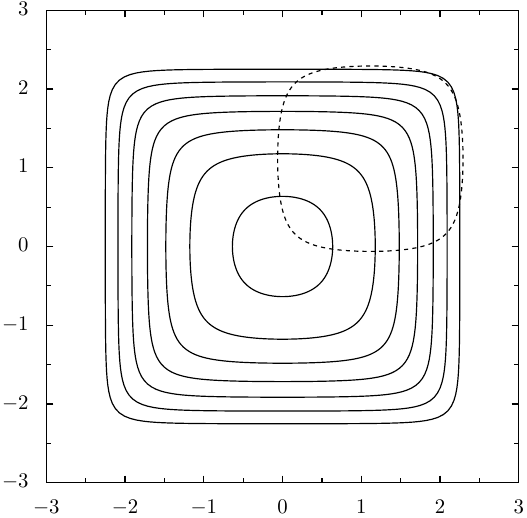}
\caption{The level sets for $f(x,y)=e^{x^2}+e^{y^2}$ (shown for function values $2.5$, $5$, $10$, $20$, $40$, $80$, and $160$) do not form pseudocircles.}
\label{fig:expx2}
\end{wrapfigure}

Figure~\ref{fig:expx2} shows an example in which $g(x)$ and $h(y)$ grow so quickly that $g'''g'>(g'')^2$ and $h'''h'>(h'')^2$ for sufficiently large $x$.
Specifically, $g'''g'=(16x^4+24x^2)e^{2x^2}$, while $(g'')^2=(16x^4+16x^2+4)e^{2x^2}$; it follows that $g'''g'>(g'')^2$ for $x>1/\sqrt 2$. One level set has been translated so that it crosses the outer level set four times, so these level sets are not pseudo\-circles.
More generally, whenever $g'''g'>(g'')^2$ over some interval of values of $x$, the level sets of $f(x,y)=g(x)+g(y)$ do not form pseudocircles: the radius of curvature at one of the tangents with slope $-1$ shrinks rather than growing for increasing values of~$f$.

Figures~\ref{fig:diamonds} and~\ref{fig:quadvor} provide a sketch of a construction for a convex function that has minimization diagrams with quadratically growing complexity. The function grows most slowly on a horizontal line through the origin, and more quickly along any other horizontal line, so that for any sites $p$ and $q$ that are not on a horizontal line, the Voronoi cell for $p$ dominates that for $q$ for points far enough away from both sites on a horizontal line through $p$. In particular the vertical line of cells to the right of the figure generates a sequence of cells with horizontal boundaries that extends, despite interruptions, through the whole figure. A more widely spaced sequence of sites at the bottom of the figure interrupts these cells, splitting them into linearly many pieces.

\begin{figure}[hb]
\centering\includegraphics[width=3in]{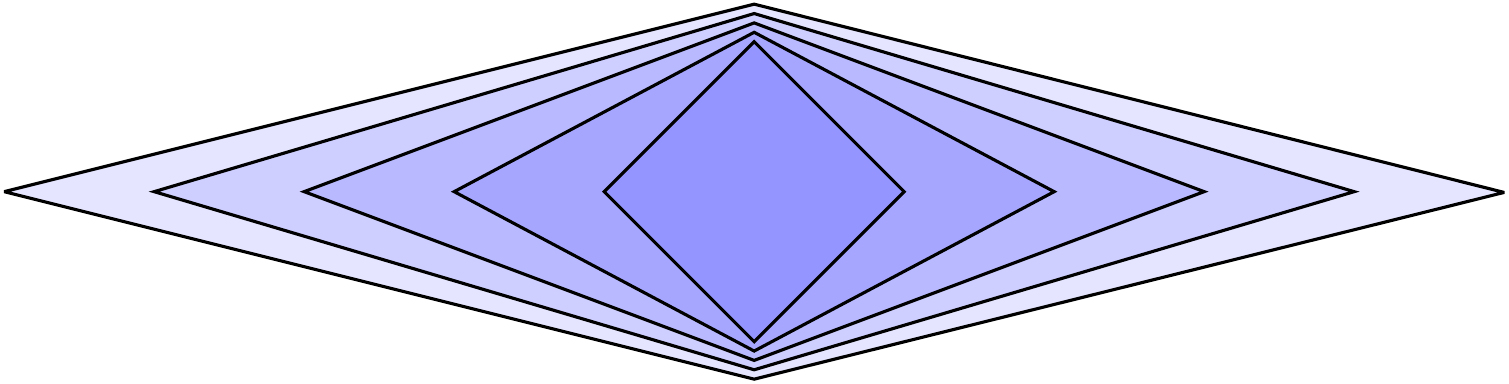}
\caption{Level sets for a convex function the minimization diagrams of which have quadratic complexity.}
\label{fig:diamonds}
\end{figure}

\begin{figure}[hb]
\centering\includegraphics[width=3.25in]{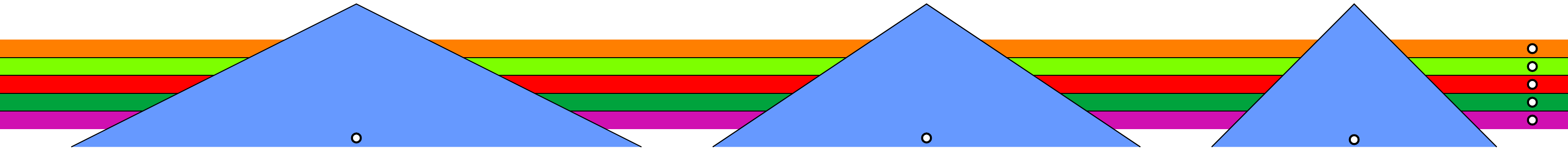}
\caption{Sketch of the structure for a minimization diagram of a convex function with quadratic complexity.}
\label{fig:quadvor}
\end{figure}

\section{Voronoi Diagrams for Smoothed Distance}

So, now that we've transformed the smoothed distance Voronoi diagram into the form of a minimization diagram for a function $f(x,y)=g(x)+h(y)$, and now that we know conditions on $g$ and $h$ that ensure that the minimization diagram to have linear complexity and be constructable efficiently, do the specific $g$ and $h$ arising in this application meet these conditions? The answer is: it depends.

First, consider the function $g(x)=\ln(e^{x} + 2+e^{-x})$, with derivatives
$g'(x) =(e^x-1)/(e^x+1)$,
$g''(x) = {2e^x}/{(1+e^x)^2}$, and
$g'''(x) = -{2e^x(e^x-1)}/{(e^x+1)^3}$.
For positive $x$, $g'''$ and therefore $g'''g'$ are negative while $(g'')^2$ is positive, so $g'''g'<(g'')^2$.
Because the function is symmetric ($g(-x)=g(x)$) the same holds for negative $x$.
And at zero, $g'''g'=0$ while $(g'')^2$ is nonzero. Therefore, $g$ meets the conditions of Theorem~\ref{thm:voronoi}.

Next, consider the function $h(y)=-\ln(1+\cos y)$. Its derivatives are
$h'(y)=\tan(y/2)=\sin y/(1+\cos y)$,
$h''(y)=1/(1+\cos y)$, and
$h'''(y)=\sin y/(1+\cos y)^2$.
Then $h'''h'=\sin^2 y/(1+\cos y)^3$, while $(h'')^2=1/(1+\cos y)^2$.
The ratio of these two values, $\sin^2 y/(1+\cos y)^2=\tan^2(y/2)$, is less than one
when $|y|<\pi/2$ and greater than one when $|y|>\pi/2$. Thus, $h$ satisfies the requirement that $h'''h'<(h'')^2$ only for $y$ in the interval $(-\pi/2,\pi/2)$. This implies that we may only apply Theorem~\ref{thm:voronoi} to smoothed distance when each Voronoi cell consists of points spanning at most a right angle with $o$ and the cell's site.

\begin{theorem}
\label{thm:smoothed}
Let $P$ be a point set such that, in the Voronoi diagram for $o$-smoothed distance, every point $q$ within the Voronoi cell for a site $p$ forms an angle $\angle poq$ that is at most $\pi/2$.  Then each cell in the Voronoi diagram is connected, and the diagram may be constructed in randomized expected time $O(n\log n)$.
\end{theorem}

\begin{proof}
To prove this, in outline, we replace the point set $P$ by its logarithmically transformed image $\{p\mid e^p\in P\}$; each point in $P$ corresponds to infinitely many transformed points, all with the same $x$ coordinate and with $y$ coordinates that differ by integer multiples of $2\pi$. Under this transformation, each point in the transformed plane is associated with a Voronoi region that the exponential function maps back to the correct Voronoi region for the associated input point in the original plane. Then, instead of using the functions $g(x)$ and $h(y)$ in the transformed plane, as defined above, we replace the function $h(y)$ by a modified function that has the same values within the interval $(-\pi/2,\pi/2)$ but that obeys the inequality $h'''h'<(h'')^2$ for larger values as well. This replacement allows us to apply Theorem~\ref{thm:voronoi} to the transformed input, and
we show that, with the assumptions stated in the theorem, it produces the same cell decomposition as the one we wish to compute. We use an efficient randomized incremental algorithm to compute the smoothed Voronoi diagram, similar to the algorithm used in Theorem~\ref{thm:voronoi}. 

Smoothed distance may be replaced by the translates of a convex function by logarithmically mapping the sites and monotonically transforming the distance values.
Thus, Voronoi diagrams for smoothed distance are closely related to minimization diagrams for this convex function. The specific relation is this: if we view the points in the plane as complex numbers,
with $o=0$, and map the finite set $P$ of sites to the infinite vertically-periodic set $L=\{q\mid e^q\in P\}$, then the exponential function forms a covering map from the minimization diagram of $f$ with respect to $L$ to the Voronoi diagram for smoothed distance of $P$. The cells in the minimization diagram are mapped many-to-one to cells in the Voronoi diagram, and edges and features in the minimization diagram are mapped to edges and features in the Voronoi diagram, etc. Although it is problematic to perform geometric algorithms on infinite point sets such as $L$, we may use our analysis to show that cells in the minimization diagram for $L$ are simply connected, while performing the randomized incremental algorithm described in the proof of Theorem~\ref{thm:voronoi} directly on the finite point set $P$.

There is a difficulty with this approach, however: our proof that the cells are simply connected does not apply, because the convex function into which we have transformed smoothed distance does not meets the requirement of Theorem~\ref{thm:voronoi} that the function $h$ satisfies the inequality $h'''h'<(h'')^2$, at least for large values of $y$. And although the assumptions of Theorem~\ref{thm:smoothed} imply that only small values of $y$ need be considered in the final Voronoi diagram, larger values may be needed at early stages of our randomized incremental construction.

Therefore, rather than computing the Voronoi diagram for smoothed distance $d_o(p,q)$ itself, we compute
the Voronoi diagram for a (non-metric) distance function $D(p,q)=g(\ln(d(p,o)/d(q,o))) + h(\angle poq)$.
Here $g(x)=\ln(e^x+2+e^{-x})$, matching the monotonically-transformed formula for smoothed distance. However, we set $h(y)$ to be a function that equals $-\ln(1+\cos y)$ for $|y|\le\pi/2$ but that is a quadratic polynomial for values of $y$ outside this range; the two quadratic polynomials (for positive and negative $y$) are chosen to have a value, first derivative, and second derivative that matches $-\ln(1+\cos y)$ at $-\pi/2$ and $\pi/2$. However, being quadratic polynomials, they have third derivative equal to zero, and therefore satisfy the requirement that $h'''h'<(h'')^2$ for the range of values of $y$ in which they define the value of $h$. The fact that the overall $h$ function has a discontinuous third derivative at $-\pi/2$ and at $\pi/2$ does not cause any problems for our analysis.

With this modified distance function, Theorem~\ref{thm:voronoi} can be used to show that the cells of the minimization diagram for $g(x)+h(y)$ with respect to the points of $L$ are simply connected, and therefore that their images, the cells of the Voronoi diagram for $D$ with respect to the points of $P$, are also connected. We may then apply a randomized incremental algorithm of the type described in Theorem~\ref{thm:voronoi} to construct the Voronoi diagram for $D$ with respect to the points of $P$. At intermediate stages of this construction, the cells of the diagram may have self-adjacencies (corresponding to boundaries between two different images in $L$ of a point in $P$) but these need not be treated any differently than any other of the bisectors in the diagram, and must vanish before the algorithm finishes.

It remains to show that each cell in the diagram for the modified distance function is equal to the corresponding cell in the diagram for the true smoothed distance $d_o$. Observe that, for any site $p$, the cell for $p$ for distance $D$ is contained within the boundaries of the cell for $p$ for distance $d_o$: by assumption, these boundaries are all within the region of the plane in which $D(p,q)$ and (the corresponding monotonic transformation of) $d_o(p,q)$ coincide, so on those boundary points $d_o(p,q)$ is accurately represented by $D(p,q)$ while the distances from $q$ to other points as measured by $D$ may be underestimates of the true value as measured by $d_o$. Thus, the cells for the Voronoi diagram for $D$ are subsets of the corresponding cells for the Voronoi diagram of $d_o$. However, since the cells for the Voronoi diagram for $D$ nevertheless cover the plane, both sets of cells coincide and the computed Voronoi diagram is correct.
\end{proof}

\section{Lloyd's Algorithm}
\label{sec:lloyd}

Evenly spaced points in Euclidean and related metric spaces have applications ranging from coding theory~\cite{Llo-ITIT-82} and color quantization~\cite{Hec-SIGGRAPH-82} to dithering (spatial halftoning) and stippling for image rendering~\cite{MacIseAnd-CG-08,Uli-87}. However, random points  typically have uneven spacing, and metric $\epsilon$-nets (maximal point sets such that no two points are closer than $\epsilon$ to each other), while more uniform, may still have varying density. A common method for improving the spacing of Euclidean point sets is \emph{Lloyd's algorithm}~\cite{Llo-ITIT-82}, a variant of the $K$-means clustering algorithm that repeatedly computes a Voronoi diagram and replaces each point by the centroid of its Voronoi cell. Its output is a \emph{centroidal Voronoi diagram}~\cite{DuFabGun-SR-99}, a well-spaced collection of points that form the centroids of their Voronoi cells.

Clarkson~\cite{Cla-UCI-08} suggested using metric $\epsilon$-nets for $o$-smoothed distance to generate well-spaced point sets with a distribution centered at $o$ that decreases exponentially with distance from $o$. For this application, one must restrict the points to an annulus centered on $o$; otherwise, an $\epsilon$-net would not have a finite number of points. As an alternative means of generating exponentially-distributed and well-spaced points in this annulus, we experimented with a variant of Lloyd's algorithm that uses smoothed distance in place of Euclidean distance in its calculations.
Specifically, rather than computing a Euclidean Voronoi diagram of the given points, we computed a Voronoi diagram for the smoothed distance. And rather than moving each point to the centroid of its cell $C$ (the point $p$ minimizing $\int_{q\in C} d(p,q)^2 dC$), we move each point to the point minimizing $\int_{q\in C} d_o(p,q)^2 dC$. Finally, in order to make the measure $dC$ of area used in the definition of the area integral transform scale-invariantly to match the symmetries of the smoothed distance, we chose a measure that is uniform not in the Euclidean plane in which the smoothed distance is defined, but rather the uniform measure in the transformed plane that has the Euclidean polar coordinates $(\log r,\theta)$ as its Cartesian coordinates.

\begin{figure}[t]
\centering
\includegraphics[width=1.5in]{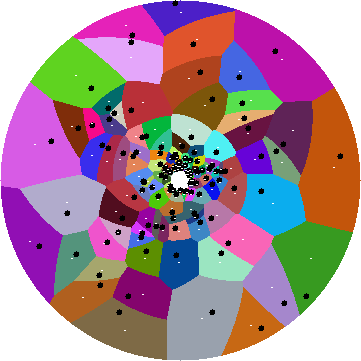}\qquad
\includegraphics[width=1.5in]{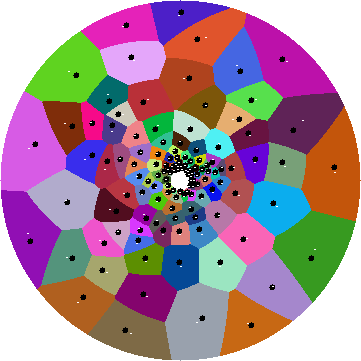}\\[0.1in]
\includegraphics[width=1.5in]{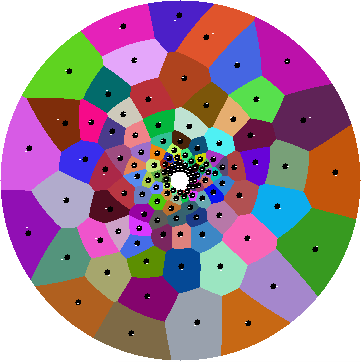}\qquad
\includegraphics[width=1.5in]{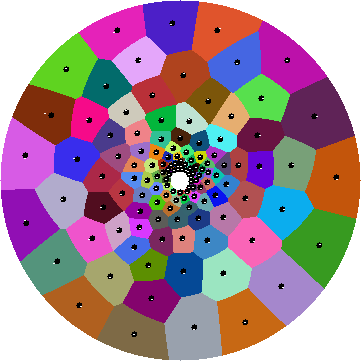}\\[0.1in]
\includegraphics[width=1.5in]{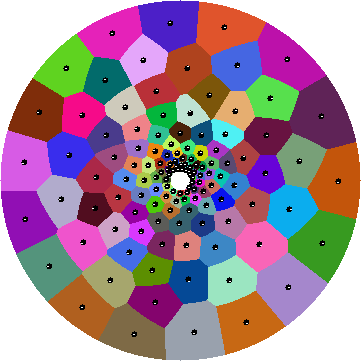}\qquad
\includegraphics[width=1.5in]{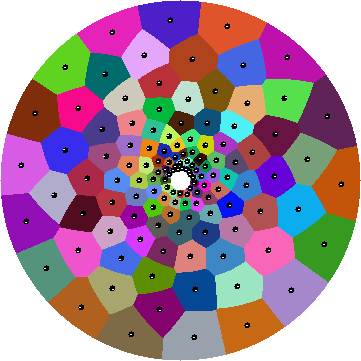}
\caption{Lloyd's algorithm on an annulus with an 18:1 ratio of inner to outer radius, for 128 exponentially-distributed random points. Top: initial configuration and one iteration. Center: after two and four iterations. Bottom: after 4, 8, and 16 iterations. The large black dots are the sites at each iteration; the smaller white spots represent the smoothed-distance Voronoi centroids to which these sites will be moved at the next iteration.}
\label{fig:lloyd}
\end{figure}

For simplicity, our implementation performs its calculations in the Euclidean plane, rasterized as a bitmap image. We compute the Voronoi diagram by finding the nearest site to each pixel of the rasterized annulus, and approximate $\int_{q\in C} d_o(p,q)^2 dC$ by $\sum_{q\in C} d_o(p,q)^2/d(q,o)^2$, where the sum is over the pixels in the Voronoi region $C$ and the $1/d(q,o)^2$ term weights each pixel by its measure in the transformed plane. The results of one run of our implementation are depicted in Figure~\ref{fig:lloyd}. We found that the iteration quickly (within two iterations) smoothed out any gross variation in the spacing of the given points, and then more slowly converged to a more ideal shape for each Voronoi cell. It was necessary to choose an initial set of points that was exponentially distributed around $o$; we placed each point by choosing $L$ uniformly within the interval $[\ln r,\ln R]$ (where $r$ and $R$ are the inner and outer radius of the annulus) and $\theta$ uniformly in $[0,2\pi]$, and then selecting the point  $(e^L\cos\theta,e^L\sin\theta)$ in the Euclidean plane. We found that if instead we selected points with uniform Euclidean measure in the annulus, too many points were placed far from $o$ and too few were placed close to $o$; Lloyd's algorithm was slow in correcting this imbalance.

\section{Conclusions}

We have identified a general condition on convex functions that causes their minimization diagrams to have linear complexity, applied this condition to Voronoi diagrams of smoothed distance and to finding the minimum dilation pair of leaves in a star network, and experimented with using a smoothing algorithm based on these Voronoi diagrams to generate evenly-spaced points exponentially distributed around a given center point.

Several directions for further research remain open:

\begin{itemize}
\item If we translate the convex function $f(x,y)$ in three dimensions rather than two by adding independent constants to the values for each point site, when does the resulting planar minimization diagram still have linear complexity? For instance, additively weighting $f(x,y)=x^2+y^2$ in this way results in a power diagram. For any convex $f(x,y)=g(x)+h(y)$, the bisectors of an additively weighted minimization diagram are still monotone curves, but we can no longer guarantee that they form pseudoline arrangements. Nevertheless it might be possible that they form minimization diagrams with connected cells.
\item Is it possible to characterize the functions that can be monotonically transformed into the form $f(x,y)=g(x)+h(y)$ with $g$ and $h$ both convex? The functions that have this form already are exactly the convex functions for which every axis-parallel rectangle has equal sums on its two pairs of opposite corners. However, if a function does not already have this form it may not be clear how to transform it into this form, as we did for smoothed distance.
\item Does our condition $g'''g'<(g'')^2$ and $h'''h'<(h'')^2$ characterize the convex functions $g$ and $h$ such that the translated level sets of $f(x,y)=g(x)+h(y)$ form pseudocircles? It does when $g=h$: If $g'''g'>(g'')^2$, then level sets of $f(x,y)=g(x)+g(y)$ have a curvature at slope 1 that grows tighter for larger circles. However the situation is less clear when $g$ and $h$ may differ.
\item In particular, does the convex function $f(x,y)=g(x)+h(y)$ coming from smoothed distance and dilation have level sets that are pseudocircles when $y>\pi/2$?
\item If not, is there some other natural distance function $d$ on the nonzero complex numbers, satisfying scale invariance $d(zp,zq)=d(p,q)$, that has simply connected Voronoi regions for all sets of two or more points in general position? Note that the most obvious choice, the Euclidean distance between $\ln p$ and $\ln q$, does not work: in general there will be a single Voronoi region containing the origin, which will not be simply connected. The replacement function used in the proof of Theorem~\ref{thm:smoothed} does not seem very natural.
\item If we define a maximization diagram in which the sites are point pairs $(p,q)$ and the function to be maximized is the dilation $d(p,q)/(d(p,o)+d(q,o))$ of $p$ and $q$ with respect to a query point $o$, our previous results~\cite{EppWor-CGTA-07} imply that this diagram has $O(n)$ cells. Are these cells simply connected?
\item Can we characterize the convex functions whose level sets are pseudocircles? For instance, as well as the functions $g(x)+h(y)$ studied here, this is also true of convex distance functions.
\item If a convex function has level sets that are pseudocircles, do its minimization diagrams automatically have simply connected cells, or is an additional condition required for this to be true?
\item To what extent can this theory be generalized to minimization diagrams in higher dimensions?
\end{itemize}

\section*{Acknowledgements}

This work was supported in part by NSF grant
0830403 and by the Office of Naval Research under grant
N00014-08-1-1015. We thank Elena Mumford for helpful comments on a draft of this paper.

\bibliographystyle{IEEEtran}
\bibliography{convex-min-diagrams}

\end{document}